\documentclass[reqno,12pt]{amsart}
\usepackage{graphicx}
\usepackage{amssymb}
\usepackage{amsfonts}
\usepackage{amsmath}

\newtheorem{theorem}{Theorem}[section]
\theoremstyle{plain}

\newtheorem{lemma}{Lemma}[section]

\newtheorem{remark}{Remark}[section]

\numberwithin{equation}{section}

\begin{document}
\title[]{Bayesian sequential estimation of the reliability of a
parallel-series system}
\author{Z. Benkamra}
\address[Z. Benkamra and M. Tlemcani]{Department of Physics, L.A.A.R \\
University Mohamed Boudiaf Oran, Algeria}
\email[Z. Benkamra]{bzohra@yahoo.com}
\email[M. Tlemcani]{mounir.tlemcani@univ-pau.fr}
\author{ M. Terbeche}
\address[M. Terbeche]{Department of Mathematics.\\
University of Oran, Algeria}
\email[M. Terbeche]{terbeche2000@yahoo.fr}
\author{M. Tlemcani}

\begin{abstract}
We give a risk-averse solution to the problem of estimating the reliability of a parallel-series system. We adopt a beta-binomial model
for components reliabilities, and assume that the total sample size for the experience is fixed. The allocation at subsystems or components
level may be random. Based on the sampling schemes for parallel and series systems separately, we propose a hybrid sequential scheme for
the parallel-series system. Asymptotic optimality of the Bayes risk associated with quadratic loss is proved with the help of martingale
convergence properties.
\end{abstract}

\maketitle
\section{Introduction}
A common approach in formulating the problem of reliability estimation relies in a multi-objective optimization problem,
i.e., maximizing the reliability estimate and minimizing the associated variance or risk. For a system with risk-averse designers or
users, the second objective is a major consideration. Because testing resources and budgets are often limited for the design,
a real difficulty lies in determining optimal sampling schemes for testing components \cite{berry}, such that the associated variance or risk
 of the system reliability estimate can be lowered by allocation. The problem reduces, therefore, to optimal allocation rules
which can be solved using dynamic programming techniques and which are costly in practice. Henceforth, and typically for large samples,
asymptotic optimality \textit{via} sequential procedures can be used as an alternative to solve the allocation problem
approximately, cf., e.g., \cite{rekab,woodroofe} and the references therein. We consider in this work the problem of estimating the reliability
of a parallel-series and/or by duality a series-parallel system, where the components reliabilities are independent Bernoulli random variables
with beta priors on their parameters and which are, themselves, independent. We assume that the total sample size for the system is fixed, but the
allocation sample sizes for subsystems or for components may be random. The system reliability looks as multivariate polynomial of components
reliabilities. In \cite {hardwick jp,hardwick j}, sequential estimation of the mean of a bivariate polynomial was discussed.
We seek for a sequential procedure which minimizes the Bayes risk subject to squared error loss. Reliability sequential schemes for a series
system were given in a frequentist framework \cite{rekab ieee}, and a bayesian framework with two components \cite{djerjour}.
We have proposed, in \cite{matcom}, a frequentist solution to the parallel-series systems, where subsystems sample sizes were fixed large
at the same order. Recently, we optimize this solution by considering only the total sample size fixed large, and the key idea is to overlap the
scheme for each parallel subsystem within the scheme for the full system at component level, cf., e.g., \cite{arxiv}.
  
In section (\ref{s1}) we discuss the bayesian framework for a parallel system. The results are used in section (\ref{s2}) to construct properly
the hybrid sequential design for the parallel-series system. The sampling schemes are shown to be asymptotically first order optimal.
\section{Reliability sequential scheme for a parallel system}
\label{s1}
\subsection{Preliminaries for a parallel system}
\label{ss1}
Consider a system $S$ of $n$ components $(1),(2),\ldots ,(n)$
connected in parallel, each one has a reliability $p_{i}$. Assume
beta priors for the Bernoulli proportions $p_{i}$ which are independent:%
\begin{equation*}
p_{i}\sim \beta \left( a_{i},b_{i}\right),
\end{equation*}
where $a_{i}$ and $b_{i}$ are some positive and known parameters.
The system reliability $p$ is%
\begin{equation*}
p=1-\prod\limits_{i=1}^{n}\left( 1-p_{i}\right)
\end{equation*}
Since a Bayesian framework is considered, then an estimator of $p$, subject
to quadratic loss, is assumed to be
\begin{equation*}
\hat{p}=1-\prod\limits_{i=1}^{n}\left( 1-\hat{p_{i}}\right),
\end{equation*}
where $\hat{p_{i}}$ is the posterior mean of the Bernoulli
proportion $p_{i}$ under beta prior, for $i=1,...,n$.
The posterior distribution of $p_{i}$ is given by $\beta
\left( a_{im_{i}},b_{im_{i}}\right) $ with parameters:%
\begin{eqnarray}
a_{im_{i}} &=&a_{i}+\sum_{k=1}^{m_{i}}x_{i}^{(k)},  \label{aimi} \\
b_{im_{i}} &=&b_{i}+m_{i}-\sum_{k=1}^{m_{i}}x_{i}^{(k)},  \label{bimi}
\end{eqnarray}%
where $x_{i}^{(k)}$ is the binary outcome of unit $(k)$ in component $%
(i)$ and $m_{i}$ the corresponding sample size. Denote by $r_{i}=a_{i}+b_{i}$, then
\begin{equation*}
\hat{p}=1-\prod\limits_{i=1}^{n}\frac{b_{im_{i}}}{a_{im_{i}}+b_{im_{i}}}%
=1-\prod\limits_{i=1}^{n}\frac{b_{i}+m_{i}-\sum%
\limits_{k=1}^{m_{i}}x_{i}^{(k)}}{m_{i}+r_{i}},
\end{equation*}%

\subsection{Asymptotics for the Bayes risk}
\label{BR}
Let $m=\sum_{i=1}^{n}{m_{i}}$ the total sample and let $\mathcal{F}_{m}$ the $\sigma $-Field
generated by $X_{1},...,X_{n}$ where $X_{i}=( X_{i}^{(1)},...,X_{i}^{(m_{i})})$.
Since the $p_{i}$ are independent random variables, the posterior expected
loss of $p$ can be written as follows:%
\begin{eqnarray*}
l\left( m_{1},...,m_{n}\right) &=&\prod\limits_{i=1}^{n}E\left[ \left(
1-p_{i}\right) ^{2}/\mathcal{F}_{m}\right] -\prod\limits_{i=1}^{n}E^{2}\left[
\left( 1-p_{i}\right) /\mathcal{F}_{m}\right] \\
&=&\prod\limits_{i=1}^{n}E\left[ \left( 1-p_{i}\right) ^{2}/\mathcal{F}_{m}%
\right] \\
&-&\prod\limits_{i=1}^{n}\left( E\left[ \left( 1-p_{i}\right) ^{2}/\mathcal{F%
}_{m}\right] -\frac{E\left[ p_{i}\left( 1-p_{i}\right) /\mathcal{F}_{m}%
\right] }{m_{i}+r_{i}}\right) \\
&=&\sum\limits_{i=1}^{n}\frac{E\left[ p_{i}\left( 1-p_{i}\right)
\prod\limits_{j\neq i}\left( 1-p_{j}\right) ^{2}/\mathcal{F}_{m}\right] }{%
m_{i}+r_{i}} \\
&+&F\left( \frac{A_{1}}{m_{1}+r_{1}},\cdots ,\frac{A_{n}}{m_{n}+r_{n}}\right),
\end{eqnarray*}

where $A_{i}$ are some random variables bounded by $1$ and $F$ is an
algebraic sum of products of at least two of its arguments. Hence, for large $m$, the Bayes
risk is%
\begin{equation}
\label{R}
R_{m}(p)=E\left[ \sum\limits_{i=1}^{n}\frac{U_{i}}{m_{i}+r_{i}}\right]
+\sum\limits_{i=1}^{n}o\left( \frac{1}{m_{i}}\right)
\end{equation}%
where $U_{i}=E\left[ V_{i}/\mathcal{F}_{m}\right]$ and
$V_{i}=p_{i}\left( 1-p_{i}\right) \prod_{j\neq i}\left( 1-p_{j}\right)^{2}$. It should be
pointed that, for all $i$ in $\{1,...,n\}$, $U_{i}={U_{i}}_{m}$ is a sequence of random variables
(the index $m$ will be omitted later to simplify the notations) depending on the allocation numbers
$m_{1},...,m_{n}$ as follows:
\begin{equation}
\label{ui}
{U_{i}}_{m}=\frac{a_{im_{i}}b_{im_{i}}}{\left( m_{i}+r_{i}\right) \left(
m_{i}+r_{i}+1\right) }\prod\limits_{j\neq i}\frac{a_{jm_{j}}\left(
a_{jm_{j}}+1\right) }{\left( m_{j}+r_{j}\right) \left( m_{j}+r_{j}+1\right) }
\end{equation}%
where $a_{im_{i}}$ and $b_{im_{i}}$ are given by expressions (\ref{aimi})
and (\ref{bimi}).

So, the Bayes risk can be approximated, for large samples, by%
\begin{equation*}
\tilde{R}_{m}(p)=E\left[ \sum\limits_{i=1}^{n}\frac{U_{i}}{m_{i}+r_{i}}\right]
\end{equation*}
It should be pointed that $m_{i}$ is treated as a random variable while the total sample size $m$ is assumed to be fixed.
Then, the approximated Bayes risk can be rewritten, thanks to Lagrange's identity, as follows: 
\begin{equation}
\label{Rtilde}
\tilde{R}_{m}(p)=\frac{E\left[ \left( \sum\limits_{i=1}^{n}\sqrt{U_{i}}\right)
^{2}+\sum\limits_{i=1}^{n-1}\sum\limits_{j=i+1}^{n}\frac{\left( \left(
m_{i}+r_{i}\right) \sqrt{U_{j}}-\left( m_{j}+r_{j}\right) \sqrt{U_{i}}%
\right) ^{2}}{\left( m_{i}+r_{i}\right) \left( m_{j}+r_{j}\right) }\right]}{%
\left( m+\sum\limits_{i=1}^{n}r_{i}\right)}
\end{equation}

\subsection{Two-stage sequential design for the parallel system}
\label{ss3}
Following the results of the previous section, sequential procedures can be developed for
the minimization of the Bayes risk when the sample size $m$ is fixed. Since $m_{i}+r_{i}$ and $m_{i}$
are of the same order for large samples, then one can choose $m_{i}$
such that%
\begin{equation*}
m_{i}\sqrt{U_{j}}=m_{j}\sqrt{U_{i}}
\end{equation*}%
for all $i,j$ in $\{1,...,n\}$; which gives the conditions%
\begin{equation*}
m_{i}=m\frac{\sqrt{U_{i}}}{\sum\limits_{j=1}^{n}\sqrt{U_{j}}}
\end{equation*}%
for all $i$ in $\{1,...,n\}$.

Assume $m$ fixed, then the allocation numbers $m_{1},...,m_{n}$ are determined by the following sequential
procedure in two stages as follows. Denote by $[x]$ the integer of $x$.

\begin{description}
\item[Stage one]
Sample $L=\left[ \sqrt{m}\right]$ units from each component $i$ and evaluate the
predictor%
\begin{equation*}
\hat{m}_{i}=m\frac{\sqrt{U_{iL}}}{\sum\limits_{j=1}^{n}\sqrt{U_{jL}}},%
\end{equation*}%
for all $i$ in $\{1,\ldots ,n\}$, where 
\begin{equation}
\label{uil}
U_{iL}=\frac{a_{iL}b_{iL}}{\left( L+r_{i}\right) \left( L+r_{i}+1\right) }%
\prod\limits_{j\neq i}\frac{a_{jL}\left( a_{jL}+1\right) }{\left(
L+r_{j}\right) \left( L+r_{j}+1\right) }
\end{equation}

\item[Stage two]
Sample $m-nL$ units of which $m_{i}-L$ are units from
component $i$ and where $m_{i}$ is the corrector defined by%
\begin{eqnarray*}
m_{i} &=&\max \left\{ L,[\hat{m}_{i}]\right\} :~i=1,\ldots ,n-1 \\
m_{n} &=&m-\sum\limits_{i=1}^{n-1}m_{i}
\end{eqnarray*}
\end{description}

\subsection{Asymptotic optimality}
\label{ss4}
Asymptotics of the Bayes risk incurred by the two-stage scheme are based
on the following lemma which is itself a consequence of the sampling procedure.

\begin{lemma}
\label{lem 411}%
The integer $m_{i}$ given by the two-stage scheme satisfies, for all
$i$ in $\{1,\ldots ,n\}$,
\begin{equation*}
\lim_{m\rightarrow +\infty }\frac{m_{i}}{m}=\frac{\sqrt{V_{i}}}{%
\sum\limits_{j=1}^{n}\sqrt{V_{j}}}, ~a.s.
\end{equation*}
\end{lemma}

\begin{proof}
Since $L=\left[ \sqrt{m}\right]$, then as $m\rightarrow +\infty $, $\frac{L}{m}\rightarrow 0$.
It follows that for large $m$,
\begin{equation*}
m_{i}=\left[m\frac{\sqrt{U_{iL}}}{\sum\limits_{j=1}^{n}\sqrt{U_{jL}}}\right]
\end{equation*}%
Moreover, as $m\rightarrow +\infty $, $L\rightarrow \infty$, and
\begin{equation*}
\frac{\sqrt{U_{iL}}}{\sum\limits_{j=1}^{n}\sqrt{U_{jL}}}\rightarrow \frac{%
\sqrt{V_{i}}}{\sum\limits_{j=1}^{n}\sqrt{V_{j}}},~a.s.,
\end{equation*}%
since, for all $i$ in $\{1,\ldots ,n\}$, $V_{i}$ is integrable and $U_{iL}=E\left[ V_{i}/A_{L}\right]$
with $A_{L}=\sigma ( X_{1}^{(1)},\ldots ,X_{1}^{(L)},\ldots
,X_{n}^{(1)},\ldots ,X_{n}^{(L)})$. The proof follows directly by the properties of the integer part at infinity.
\end{proof}

\begin{theorem}
\label{th mrm}Let $R_{m}(p)$ the Bayes risk incurred by the two-stage scheme, then%
\begin{equation*}
\lim_{m\rightarrow +\infty }mR_{m}(p)=E\left[ \left( \sum\limits_{i=1}^{n}%
\sqrt{V_{i}}\right) ^{2}\right]
\end{equation*}
\end{theorem}

\begin{proof}
A first consequence of lemma (\ref{lem 411}) is
\begin{equation}
\label{p0}
\lim_{m\rightarrow +\infty }mR_{m}(p)=\lim_{m\rightarrow +\infty }m\tilde{R}_{m}(p),
\end{equation}
where $\tilde{R}_{m}(p)$ is the asymptotic Bayes risk given by the expression (\ref{Rtilde}).

Hence, the proof follows if one shows that
\begin{equation}
\label{p1}
\lim_{m\rightarrow +\infty }E\left[ \left( \sum\limits_{i=1}^{n}\sqrt{U_{i}}%
\right) ^{2}\right] =E\left[ \left( \sum\limits_{i=1}^{n}\sqrt{V_{i}}\right)
^{2}\right]
\end{equation}%
and
\begin{equation}
\label{p2}
\lim_{m\rightarrow +\infty }E\left[ \frac{\left( \left( m_{i}+r_{i}\right) 
\sqrt{U_{j}}-\left( m_{j}+r_{j}\right) \sqrt{U_{i}}\right) ^{2}}{\left(
m_{i}+r_{i}\right) \left( m_{j}+r_{j}\right) }\right] =0.
\end{equation}%

Since $U_{i}=E\left[ V_{i}/\mathcal{F}_{m}\right]$ and
$V_{i}\in L^{1}$, the martingale $U_{i}$ is uniformly integrable
and consequently, as $m\rightarrow +\infty$, $U_{i}\rightarrow V_{i}$ in $L^{1}$ and almost surely,
for all $i$ in $\{1,...,n\}$. Martingales properties yield $\sqrt{U_{i}}\rightarrow \sqrt{V_{i}}~in~L^{2}$,
and the first identity (\ref{p1}) follows.

For the second equality (\ref{p2}), lemma (\ref{lem 411}) implies that%
\begin{equation*}
\frac{\left( \left( m_{i}+r_{i}\right) \sqrt{U_{j}}-\left(
m_{j}+r_{j}\right) \sqrt{U_{i}}\right) ^{2}}{\left(
m_{i}+r_{i}\right) \left( m_{j}+r_{j}\right) }
\end{equation*}%
converges to zero in probability, as $m\rightarrow+\infty$, for all $i,j$ in $\{1,...,n\}$.
Hence, it will be sufficient to show that this sequence is uniformly integrable.
This may be true if it is bounded by an integrable random variable. So,%
\begin{equation*}
\frac{\left( \left( m_{i}+r_{i}\right) \sqrt{U_{j}}-\left(
m_{j}+r_{j}\right) \sqrt{U_{i}}\right) ^{2}}{\left(
m_{i}+r_{i}\right) \left( m_{j}+r_{j}\right) }\leq \frac{%
m_{i}+r_{i}}{m_{j}+r_{j}}U_{j}+\frac{m_{j}+r_{j}}{
m_{i}+r_{i} }U_{i},
\end{equation*}%
\begin{equation*}
\frac{m_{i}+r_{i}}{m_{j}+r_{j}}U_{j}\leq \frac{m}{m_{j}}%
U_{j}+r_{i}U_{j}.
\end{equation*}%
Since, for large $m$,%
\begin{equation*}
\frac{m}{m_{j}}\leq \frac{\sum\limits_{i=1}^{n}\sqrt{U_{iL}}}{%
\sqrt{U_{jL}}}
\end{equation*}%
then%
\begin{equation*}
\frac{m}{m_{j}}\leq \max_{l\geq 0}\left( \sum\limits_{i=1}^{n}\sqrt{%
U_{il}U_{jl}}\right)
\end{equation*}%
Applying Doob's inequality to the right hand side of this last inequality,
one shows that this term is integrable since, for all $j$ in $\{1,...,n\}$, $U_{j}$ is bounded.

Similarly, one can do the same with the term%
\begin{equation*}
\frac{m_{j}+r_{j}}{m_{i}+r_{i}}U_{i},
\end{equation*}%
and the identity (\ref{p2}) holds.

The proof of the theorem follows from (\ref{p0}), (\ref{p1}) and (\ref{p2}).
\end{proof}
\begin{remark}
\label{rem1}
Using the duality between the parallel and the series configurations, all the results
obtained for a parallel system can be adapted straightforwardly to a series system.
More precisely, one has just to alternate the roles  between components reliabilities
$(p_{i})$, respectively system reliability $(p)$, and probabilities of failure $(q_{i}=1-p_{i})$,
respectively $(q=1-p)$; the Bayes risk remaining the same.
\end{remark}
\section{Extension to a parallel-series system}
\label{s2}
\subsection{Preliminaries}
\label{ss21}
We consider now a parallel-series system $S$ of $n$ subsystems $S_{1},...,S_{n}$ connected
in series, each subsystem $S_{i}$ has a reliability $p_{i}$ and contains $n_{i}$
components $C_{i1},...,C_{in_{i}}$ connected in parallel, each component $C_{ij}$ has a reliability $p_{ij}$.

As in section (\ref{s1}), beta priors are assumed on the proportions $p_{ij}$,%
\begin{equation*}
p_{ij}\sim \beta \left( a_{ij},b_{ij}\right),
\end{equation*}
where $a_{ij}$ and $b_{ij}$ are positive and known parameters. Assuming independence across and within
populations, reliabilities estimates are shortly summarized as bellow.
\begin{eqnarray}
\hat{p}&=&\prod\limits_{i=1}^{n} \hat{p_{i}}\\
&=&\prod\limits_{i=1}^{n}\left( 1-\prod\limits_{j=1}^{n_{i}}\left(1-\hat{p_{ij}}\right)\right)
\end{eqnarray}
where $\hat{p_{ij}}$ is the posterior mean of the Bernoulli
proportion $p_{ij}$. The posterior distribution of $p_{ij}$ is given by $\beta
\left( a_{ijm_{ij}},b_{ijm_{ij}}\right) $ with 
\begin{eqnarray}
a_{ijm_{ij}} &=&a_{ij}+\sum_{k=1,m_{ij}}x_{ij}^{(k)},  \label{aijmij} \\
b_{ijm_{ij}} &=&b_{ij}+m_{ij}-\sum_{k=1,m_{ij}}x_{ij}^{(k)},  \label{bijmij}
\end{eqnarray}%
where $x_{ij}^{(k)}$ is the binary outcome of unit $(k)$ in component $%
C_{ij}$ and $m_{ij}$ the corresponding sample size. 
\subsection{The Bayes risk}
\label{ss22}
The integer $\sum_{j}m_{ij}=m_{i}$ is the sample size in the subsystem $S_{i}$, while $\sum_{i}m_{i}=m$
is the total sample size in the system. Denote by $\mathcal{\tilde{F}}_{m}$ the $\sigma $-field generated
by $\left( X_{1},...,X_{n}\right) $, where $X_{i}=( X_{i1},...,X_{in_{i}})$ and
$X_{ij}=(X_{ij}^{(1)},...,X_{ij}^{(m_{ij})})$.

Assuming quadratic loss, typically $\hat{p}$, $\hat{p_{i}}$, respectively $\hat{p_{ij}}$ are Bayes
estimators of system, subsystems, respectively components reliabilities. The Bayes risk $R_{m}(p)$ is the mean with respect
to evidence of the posterior variance. With the help of independence, one can write
\begin{eqnarray*}
Var\left( p/\mathcal{\tilde{F}}%
_{m}\right) &=&E\left[ \prod\limits_{i=1}^{n}p_{i}^{2}/\mathcal{\tilde{F}}%
_{m}\right] -E^{2}\left[ \prod\limits_{i=1}^{n}p_{i}/\mathcal{\tilde{F}}_{m}%
\right] \\
&=&\prod\limits_{i=1}^{n}E\left[ p_{i}^{2}/\mathcal{\tilde{F}}_{m}\right]
-\prod\limits_{i=1}^{n}E^{2}\left[ p_{i}/\mathcal{\tilde{F}}_{m}\right]\\
&=&\prod\limits_{i=1}^{n}E\left[ p_{i}^{2}/\mathcal{\tilde{F}}%
_{m}\right] -\prod\limits_{i=1}^{n}\left( E\left[ p_{i}^{2}/\mathcal{\tilde{F%
}}_{m}\right] -Var\left[ p_{i}/\mathcal{\tilde{F}}_{m}\right] \right) .
\end{eqnarray*}%
Hence, the Bayes risk is given by
\begin{equation}
\label{br1}
R_{m}(p)=\prod\limits_{i=1}^{n}E\left[ p_{i}^{2}\right] -\prod\limits_{i=1}^{n}%
\left( E\left[ p_{i}^{2}\right] -R_{m_{i}}\left( p_{i}\right) \right) ,
\end{equation}%
where $R_{m_{i}}\left( p_{i}\right) $ is the Bayes risk corresponding to the subsystem $S_{i}$.
\subsection{Asymptotics induced by the two-stage at component level}
\label{ss23}
For each $m_{i}$ fixed (large), if one determines the partition $\{m_{ij},~j=1,...,n_{i}\}$ according to the
two-stage design defined in the previous section for the parallel subsystem $S_{i}$, then the corresponding
Bayes risk can be written, thanks to theorem (\ref{th mrm}), as follows.
\begin{equation}
R_{m_{i}}\left( p_{i}\right) =\frac{B_{i}}{m_{i}}+o\left( \frac{1}{m_{i}}\right),
\end{equation}
with $B_{i}$ constant,
\begin{equation*}
B_{i}=E\left[ \left( \sum\limits_{j=1}^{n_{i}}\sqrt{V_{ij}}\right) ^{2}\right],
\end{equation*}
and $V_{ij}$ are random variables defined similarly as in subsection (\ref{BR}).
It follows that
\begin{eqnarray*}
R_{m}(p) &=&\prod\limits_{i=1}^{n}E\left[ p_{i}^{2}\right] -\prod\limits_{i=1}^{n}%
\left( E\left[ p_{i}^{2}\right] -\frac{B_{i}}{m_{i}}\right)
+\sum\limits_{i=1}^{n}o\left( \frac{1}{m_{i}}\right) \\
&=&\tilde{R}_{m}(p)+\sum\limits_{i=1}^{n}o\left( \frac{1}{m_{i}}\right),
\end{eqnarray*}
where
\begin{equation*}
\tilde{R}_{m}(p)=E\left[ \sum\limits_{i=1}^{n}\frac{B_{i}w_{i}}{m_{i}}\right],
\end{equation*}%
is the asymptotic Bayes risk, and $w_{i}=E[ Z_{i}/\mathcal{\tilde{F}}_{m}]$, where
\begin{equation*}
Z_{i}=\prod\limits_{l=1\ {l\neq i}}^{n}p_{l}^{2}
\end{equation*}
A calculus based on posterior distributions gives
\begin{eqnarray}
w_{i} &=&\prod\limits_{l\neq i}\left[ 1+\prod\limits_{j=1}^{n_{l}}\frac{%
\left( b_{lj}+m_{lj}-\sum\limits_{k=1}^{m_{lj}}x_{lj}^{(k)}\right) \left(
b_{lj}+m_{lj}-\sum\limits_{k=1}^{m_{ij}}x_{lj}^{(k)}+1\right) }{\left(
a_{lj}+b_{lj}+m_{lj}\right) \left( a_{lj}+b_{lj}+m_{lj}+1\right) }\right. \nonumber \\ 
&&\left. -2\prod\limits_{j=1}^{n_{l}}\frac{\left(
b_{lj}+m_{lj}-\sum\limits_{k=1}^{m_{lj}}x_{lj}^{(k)}\right) }{\left(
a_{lj}+b_{lj}+m_{lj}\right) \left( a_{lj}+b_{lj}+m_{lj}+1\right) }\right]
\label{wi}
\end{eqnarray}%

\subsection{A hybrid sequential design for the parallel-series system}

It follows from Lagrange's identity that,
\begin{equation}
\tilde{R}_{m}(p)=\frac{E\left[ \left( \sum\limits_{i=1}^{n}\sqrt{B_{i}w_{i}}%
\right) ^{2}+\sum\limits_{i=1}^{n-1}\sum\limits_{k=i+1}^{n}\dfrac{\left( m_{i}%
\sqrt{B_{k}w_{k}}-m_{k}\sqrt{B_{i}w_{i}}\right) ^{2}}{m_{i}m_{k}}\right]}{m}
\label{rp1}
\end{equation}
Hence, for $m$ fixed, and according to this last identity, one can allocate sequentially the sample
size $m_{i}$ at subsystem level such that, for all $i,k$ in $\{1,...,n\}$, the conditions 
\begin{equation*}
m_{i}\sqrt{B_{k}w_{k}}=m_{k}\sqrt{B_{i}w_{i}}
\end{equation*}%
are satisfied, or equivalently,  for all $i$ in $\{1,...,n\}$,%
\begin{equation}
\label{mirem1}
m_{i}=m\frac{\sqrt{B_{i}w_{i}}}{\sum\limits_{k=1}^{n}\sqrt{B_{k}w_{k}}}.
\end{equation}%

\begin{remark}
\label{rem2}
Following the remark (\ref{rem1}), the criteria (\ref{mirem1}) is similar but not identical to the rule
expected for allocating $m_{i}$ as in a pure series system where each subsystem is reduced to a single component.
The difference lies in the weighting coefficients $B_{i}$ (the rates of convergence of the risks $R_{m_{i}}(p_{i})$)
which allow to take into account the estimation cost for each parallel subsystem reliability at component level.
\end{remark}

As a result, we propose the following hybrid two-stage scheme.
Let $L=[\sqrt{m}]$, respectively $\tilde{L}=[\sqrt{L}]$, the initial sample size for each subsystem
$S_{i}$, respectively each component $C_{ij}$.
\begin{description}
\item[Stage one]Sample $\tilde{L}$ units in each component of the system.
\begin{itemize}
\item[(i)] Evaluate $w_{i}$ (denoted by $\tilde{w}_{i}$) according to the formula
(\ref{wi}) with $m_{i}=L$ and $m_{ij}=\tilde{L}$.
\item[(ii)] Calculate the predictor at subsystem level,
\begin{equation*}
\tilde{m}_{i}=\left[m\frac{\sqrt{B_{i}\tilde{w}_{i}}}{\sum\limits_{j=1}^{n}\sqrt{
B_{j}\tilde{w}_{j}}}\right].
\end{equation*}
\end{itemize}
\item[Stage two]
Sample $m-nL$ units more for which $m_{i}-L$ are in the subsystem $S_{i}$, where%
\begin{eqnarray*}
m_{i} &=&\max \left\{ L,\tilde{m}_{i}\right\} :~i=1,\ldots ,n-1, \\
m_{n} &=&m-\sum\limits_{i=1}^{n-1}m_{i},
\end{eqnarray*}%
and calculate the sample size $m_{ij}$ at component level, according to the second stage of the
sequential scheme defined in subsection (\ref{ss3}), i.e.,
\begin{eqnarray*}
m_{ij} &=&\max \left\{ \tilde{L},\left[m_{i}\frac{\sqrt{U_{ij\tilde{L}}}}{%
\sum\limits_{k=1}^{n_{i}}\sqrt{U_{ik\tilde{L}}}}\right]\right\} :~j=1,\ldots
,n_{i}-1, \\
m_{in_{i}} &=&m_{i}-\sum\limits_{j=1}^{n_{i}-1}m_{ij},
\end{eqnarray*}%
where $U_{ij\tilde{L}}$ are evaluated properly by the relation (\ref{uil}).
\end{description}

\subsection{First order optimality}
As in the parallel case, cf. section (\ref{s1}), we have the following convergence result.
\begin{theorem}
\label{th-21}
The Bayes risk $R_{m}(p)$ incurred by the hybrid sequential design satisfies
\begin{equation*}
\lim_{m\rightarrow +\infty }m.R_{m}(p)=E\left[ \left( \sum\limits_{i=1}^{n}%
\sqrt{B_{i}Z_{i}}\right) ^{2}\right] .
\end{equation*}
\end{theorem}
Following the arguments used in lemma (\ref{lem 411}), and the hybrid
sequential scheme, one can show that $m_{i}$ and $m_{ij}$ are of the same order at
infinity, for all $i,j$. The following lemma yields a similar comparison of $m$ and $m_{i}$.    
\begin{lemma}
\label{lem 421}%
The integer $m_{i}=\sum_{j=1,n_{i}}m_{ij}$ given by the hybrid sequential scheme satisfies,
for all $i$ in $\{1,\ldots ,n\}$,
\begin{equation*}
\lim_{m\rightarrow +\infty }\frac{m_{i}}{m}=\frac{\sqrt{B_{i}Z_{i}}}{%
\sum\limits_{k=1}^{n}\sqrt{B_{k}Z_{k}}}, ~a.s.
\end{equation*}
\end{lemma}

\begin{proof}
Similar to the proof of lemma (\ref{lem 411}), since the coefficients $B_{i}$ are constant. The key point is
that $w_{i}$ converges almost surely to $Z_{i}$ (integrable), as $m\rightarrow +\infty$, i.e., as $L$ or $\tilde{L} \rightarrow +\infty$. 
\end{proof}

\begin{proof}[Proof of theorem (\ref{th-21})]
As a first consequence of lemma (\ref{lem 421}),
\begin{eqnarray*}
\lim_{m\rightarrow +\infty }m.R_{m}(p)&=&\lim_{m\rightarrow +\infty }m.\tilde{R}_{m}(p)\\ 
&=&\lim_{m\rightarrow +\infty }E\left[
\left( \sum\limits_{i=1}^{n}\sqrt{B_{i}w_{i}}\right) ^{2}\right]  \\
&&+\sum\limits_{i=1}^{n-1}\sum\limits_{k=i+1}^{n}\lim_{m\rightarrow +\infty
}E\left[ \frac{\left( m_{i}\sqrt{B_{k}w_{k}}-m_{k}\sqrt{B_{i}w_{i}}\right)
^{2}}{m_{i}m_{k}}\right]
\end{eqnarray*}%
The proof follows if one shows that
\begin{equation}
\lim_{m\rightarrow +\infty }E\left[ \left( \sum\limits_{i=1}^{n}\sqrt{%
B_{i}w_{i}}\right) ^{2}\right] =E\left[ \left( \sum\limits_{i=1}^{n}\sqrt{%
B_{i}Z_{i}}\right) ^{2}\right]   \label{rtp1}
\end{equation}%
and 
\begin{equation}
\lim_{m\rightarrow +\infty }E\left[ \frac{\left( m_{i}\sqrt{B_{k}w_{k}}-m_{k}%
\sqrt{B_{i}w_{i}}\right) ^{2}}{m_{i}m_{k}}\right] =0.  \label{rtp2}
\end{equation}%
Similarly, as for theorem (\ref{th mrm}),  equality (\ref{rtp1}) follows from the uniform integrability
of the martingale $w_{i}=E\left[ Z_{i}/\mathcal{\tilde{F}}_{m}\right]$ and martingales properties which yield
a convergence in $L^{2}$ of $\sqrt{w_{i}}$ to $\sqrt{Z_{i}}$, as $m\rightarrow +\infty $.
All the same, the second identity is a consequence of lemma (\ref{lem 421}), the hybrid sequential
scheme, and uniform integrability of the sequence in (\ref{rtp2}) which follows with the help of Doob's inequality.
\end{proof}

\section{Conclusion}
The hybrid sequential scheme was constructed, based on the two-stage sampling scheme for each parallel subsystem
at component level and the sampling scheme for the series structure at subsystem level. The first order optimality
was obtained, mainly, by the martingale convergence properties and Doob's inequality. With minor changes, the series-parallel
systems may be treated similarly, using duality. Namely, the techniques discussed here can be tediously adapted for complex systems
involving a multi-criteria optimization problem under a set of constraints such as risk, system weight,
cost, performance and others.


\end{document}